\documentclass[11pt]{amsart}

\usepackage[colorlinks = true,
            linkcolor = red,
            urlcolor  = magenta,
            citecolor = black,
            anchorcolor = blue]{hyperref}
\usepackage{color}
\usepackage[shortalphabetic,initials,nobysame]{amsrefs}
\usepackage{upgreek}
\usepackage{tikz}
\usepackage[applemac]{inputenc} 
\usetikzlibrary{arrows}
\usepackage{xcolor}
\usepackage[all]{xy}
\usepackage{graphicx}
\usepackage{wrapfig}
\usepackage{yfonts}
\usepackage{graphicx}
\usepackage{amssymb}
\usepackage{epstopdf}
\usepackage{mathrsfs}
\usepackage[vcentermath]{youngtab}
\usepackage{bm}

\usepackage[mathscr]{euscript}
 \let\mathscr\relax
\usepackage[scr]{rsfso}

\makeatletter \@addtoreset{equation}{section} \makeatother

\renewcommand{\eprint}[1]{\href{https://arxiv.org/abs/#1}{#1}}

\BibSpec{article}{%
    +{}  {\PrintAuthors}                {author}
    +{,} { \textit}                     {title}
     +{,} {}                            {note}
    +{.} { }                            {part}
    +{:} { \textit}                     {subtitle}
    +{,} { \PrintContributions}         {contribution}
    +{.} { \PrintPartials}              {partial}
    +{,} { }                            {journal}
    +{}  { \textbf}                     {volume}
    +{}  { \PrintDatePV}                {date}
    +{,} { \issuetext}                  {number}
    +{,} { \eprintpages}                {pages}
    +{,} { }                            {status}
    +{,} { \DOI}                   {doi}
    +{,} { \eprint}        {eprint}
      +{,} {\publisher}                {publisher}
    +{,} { \address}                   {address}
    +{}  { \parenthesize}               {language}
    +{}  { \PrintTranslation}           {translation}
    +{;} { \PrintReprint}               {reprint}
    +{.} {}                             {transition}
    +{}  {\SentenceSpace \PrintReviews} {review}
}
\def\ee{\end{eqnarray}}

\makeatletter
\newcommand{\ellSN}{\mathop{\operator@font sn}\nolimits}
\newcommand{\ellCN}{\mathop{\operator@font cn}\nolimits}
\newcommand{\ellDN}{\mathop{\operator@font dn}\nolimits}
\newcommand{\ellAM}{\mathop{\operator@font am}\nolimits}
\newcommand{\ellK}{\mathop{\smash{\operator@font K}\vphantom{a}}\nolimits}
\newcommand{\ellE}{\mathop{\smash{\operator@font E}\vphantom{a}}\nolimits}
\makeatother

\ifx\genfrac\sdflkaj

\else

\fi

\newcommand{\beq}{\begin{equation}}
\newcommand{\eeq}{\end{equation}}

\newtheorem{theorem}{Theorem}[section]

\newtheorem{proposition}[theorem]{Proposition}

\newtheorem{definition}[theorem]{Definition}

\makeatletter
\def\mr@ignsp#1 {\ifx\:#1\@empty\else #1\expandafter\mr@ignsp\fi}%
\newcommand{\multiref}[1]{\begingroup
\xdef\mr@no@sparg{\expandafter\mr@ignsp#1 \: }%
\def\mr@comma{}%
\@for\mr@refs:=\mr@no@sparg\do{\mr@comma\def\mr@comma{,}\ref{\mr@refs}}%
\endgroup}
\makeatother

\newcommand{\hypref}[2]{\ifx\href\asklfhas #2\else\href{#1}{#2}\fi}

\newcommand{\figref}[1]{Fig.~\multiref{#1}}
\renewcommand{\eqref}[1]{(\multiref{#1})}

\def\[{\begin{equation}}
\def\]{\end{equation}}
\def\<{\begin{eqnarray}}
\def\>{\end{eqnarray}}

\ifx\href\asklfhas\newcommand{\href}[2]{#2}\fi

\title{Double Inozemtsev Limits of the Quantum DELL System}

\author{Alexander Gorsky}
\address[Alexander Gorsky]{\newline
          Institute for Information Transmission Problems, \newline
          Russian Academy of Sciences, \newline
          Moscow, Russia \newline
          Email: \href{mailto:shuragor@mail.ru}{shuragor@mail.ru}
          \newline \href{http://iitp.ru/en/users/3804.htm}{http://iitp.ru/en/users/3804.htm}}

\author[P. Koroteev]{Peter Koroteev}
\address{
Department of Mathematics,
University of California,
Berkeley, CA 94720, USA
and
Rutgers University, Piscataway, NJ, 08854, USA\newline
 Email: \href{mailto:pkoroteev@math.berkeley.edu}{pkoroteev@math.berkeley.edu}\newline
 \href{https://math.berkeley.edu/~pkoroteev/}{https://math.berkeley.edu/$\sim$pkoroteev}
}

\author{Olesya Koroteeva}
\address[Olesya Koroteeva]{\newline
School of Physics and Astronomy,\newline
University of Minnesota\newline
Minneapolis MN 55455\newline
United States of America\newline
Email: \href{mailto:koroteeva@physics.umn.edu}{koroteeva@physics.umn.edu}\newline
\href{https://www.physics.umn.edu/people/koroteeva.html}{https://www.physics.umn.edu/people/koroteeva.html}}

\author{Shamil Shakirov}
\address[Shamil Shakirov]{\newline
University of Geneva,\newline
Department of Mathematics,\newline
1211 Geneva 4, Switzerland\newline
Email: \href{mailto:Shamil.Shakirov@unige.ch}{Shamil.Shakirov@unige.ch}}

\date{\today}

\begin{document}

\begin{abstract}
In this letter we study various Inozemtsev-type limits of the quantum double elliptic (DELL) system when both elliptic parameters are sent to zero at different rates, while the coupling constant is sent to infinity, such that a certain combination of the three parameters is kept fixed. We find a regime in which such double Inozemtsev limit of DELL produces the elliptic Ruijsenaars-Schneider (eRS) Hamiltonians albeit in an unconventional normalization. We discuss other double scaling limits and anisotropic scaling of coordinates and momenta. In addition we provide a formal expression for the eigenvalues of the eRS Hamiltonians solely in terms of their eigenfunctions. 
\end{abstract}

\maketitle

\section{The DELL System}
The double elliptic model (DELL) is a remarkably sophisticated integrable system which has been a point of interest of many researchers during the last couple of decades. Since the advent of the Seiberg-Witten solution \cite{Seiberg:1994rs,Seiberg:1994aj} of the low energy sector of $\mathcal{N}=2$ supersymmetric gauge theories in four dimensions and its connection to classical  integrable systems  \cite{Gorsky:1993dq,Gorsky:1995zq,Martinec:1995by,Donagi:1995cf} it became clear that there must exist integrable systems which are dual to five and six dimensional gauge theories compactified on a circle and on a torus respectively. These generalization have been later studied in the literature (see \cite{Koroteev:2019tz} for review).

The classical DELL model was introduced in \cite{Hollowood:2003cv,Braden:2001yc,Braden:2003gv} using six dimensional gauge theories with eight supercharges compactified on a torus. It was shown how to construct the spectral curve, or, equivalently, the Seiberg-Witten curve of the  dual 6d theory, from M-theory. The two elliptic parameters of DELL, $p$ and $w$ are related to the gauge coupling of the 6d theory and to the elliptic modulus of the compactification torus respectively.

The quantization is an art as well as a science and thus it requires novel ideas. Nekrasov and Shatashvili \cite{Nekrasov:2009rc} suggested  a physics approach which utilizes the duality between integrable systems and gauge theories with eight supercharges. The key idea is to deform the gauge theory by the Omega-background along one of the complex directions of the Euclidian $\mathbb{R}^4$ worldvolume of the theory. The corresponding equivariant parameter becomes the Planck's constant of the dual quantum integrable system.

In \cite{Bullimore:2015fr} it was demonstrated, and later proved by Nekrasov (in the case of the four-dimensional $U(2)$ theory), see 
\cite{Nekrasov:2017ab,Nekrasov:2017aa}, how to find a formal spectrum of quantum elliptic integrable systems of Calogero or Ruijsenaars type. The wave functions of the quantum Hamiltonians are the supersymmetric partition functions in the presence of codimension-two defects of the dual 4d and 5d theories with adjoint matter respectively. The corresponding eigenfunctions are therefore represented by vacuum expectation values of local chiral observables (in 4d) and Wilson lines (in 5d). Later in  \cite{Koroteev:2019tz} this approach was extended to 6d theories.

\subsection{Overview of Quantum DELL System}
First we review the basics of the quantum double elliptic system which was discovered by two of the authors \cite{Koroteev:2019tz} and further developed in \cite{Grekov:2021zqq,Grekov:2020jpk}. (see also  \cite{Braden:1999aj,Mironov:1999vi,Aminov:2013asa,Aminov:2014wra,Aminov:2016ruk,Aminov:2017dud, Fock:1999ae, Braden:2001yc} for different approach to double elliptic models).

The DELL Hamiltonians for $N$ particles read
\begin{equation}
\mathcal{H}_a =\mathcal{O}_0^{-1} \mathcal{O}_a\,, \qquad a = 1, \ldots, N-1\,,
\label{eq:DELLHAMs}
\end{equation}
where operators $\mathcal{O}_0,\mathcal{O}_1,\dots,\mathcal{O}_{N-1}$ are Fourier modes of the following current
\begin{equation}
\mathcal{O}(z) \ = \ \sum\limits_{n \in {\mathbb Z}} \ \mathcal{O}_n \ z^n \ = \ \sum\limits_{n_1, \ldots, n_N = -\infty}^{\infty} \ (-z)^{\sum n_i} \ w^{\sum \frac{n_i(n_i - 1)}{2}} \ \prod\limits_{i < j} \theta\left( t^{n_i - n_j} \frac{x_i}{x_j}\Big\vert p \right) \  p_1^{n_1} \ldots p_N^{n_N}\,.
\label{eq:Ocurrent}
\end{equation}
In the above formula $\theta(x|p)$ is the odd theta function\footnote{The use of the odd theta function is more preferable to study the Inosemtsev limit; we changed our conventions from \cite{Koroteev:2019tz}.}
\begin{equation}\label{eq:ThetaOdd}
\theta(x|p)=(x^{\frac{1}{2}}-x^{-\frac{1}{2}})\prod_{i=1}^\infty (1-xp^i)(1-x^{-1}p^i)\,,
\end{equation}
$t$ is the exponentiated coupling constant, $z$ is an auxiliary counting parameter, and the canonically conjugate position and momentum operators obeying canonical q-commutation relation $x_i p_j=q^{\delta_{ij}}p_jx_j$ which act on functions of positions as follows
\begin{equation}
x _i f(x_1, \ldots, x_N) = x_i f(x_1, \ldots, x_N), \qquad p  _i f(x_1, \ldots, x_N) = f(x_1, \ldots, q x_i, \ldots, x_N)\,.
\end{equation}

The eigenvalue problem, which as of this writing, is a mathematical conjecture, states that a properly normalized equivariant elliptic genus of the affine Laumon space in the Nekrasov-Shatashvili limit \cite{Nekrasov:2009rc} is the eigenfunction of the quantum DELL Hamiltonians
\begin{equation}
\widehat{\mathcal{H}}_n \mathscr{Z}^{6d/4d}_{\text{inst}}(w,p,\textbf{x}) = \lambda_n(\textbf{a},w,p) \mathscr{Z}^{6d/4d}_{\text{inst}}(w,p,\textbf{x})\,.
\end{equation}
We refer the reader to \cite{Koroteev:2019tz} for more details and references. In this paper we discuss double scaling limits of DELL of Inozemtsev type \cite{inozemtsev1989}, in which the coupling $t\to\infty$ while the elliptic parameters $p$ and $w$ go to zero in the presence of additional scaling of the coordinates and momenta.

\subsection{DELL-RS-Calogero Hierarchy}
The DELL system lives on top of the hierarchy of integrable many-body systems and all other known models (without spin degrees of freedom, we shall comment on spin-DELL later) -- Calogero-Moser-Sutherland (CMS) and Ruijsenaars-Schneider (RS) systems can be obtained by decoupling certain parameters in DELL Hamiltonians, see \figref{fig:theiteptable}\footnote{In math literature (part of) this diagram is sometimes referred to as Etingof diamond.} .
\begin{figure}[h]
\includegraphics[scale=0.45]{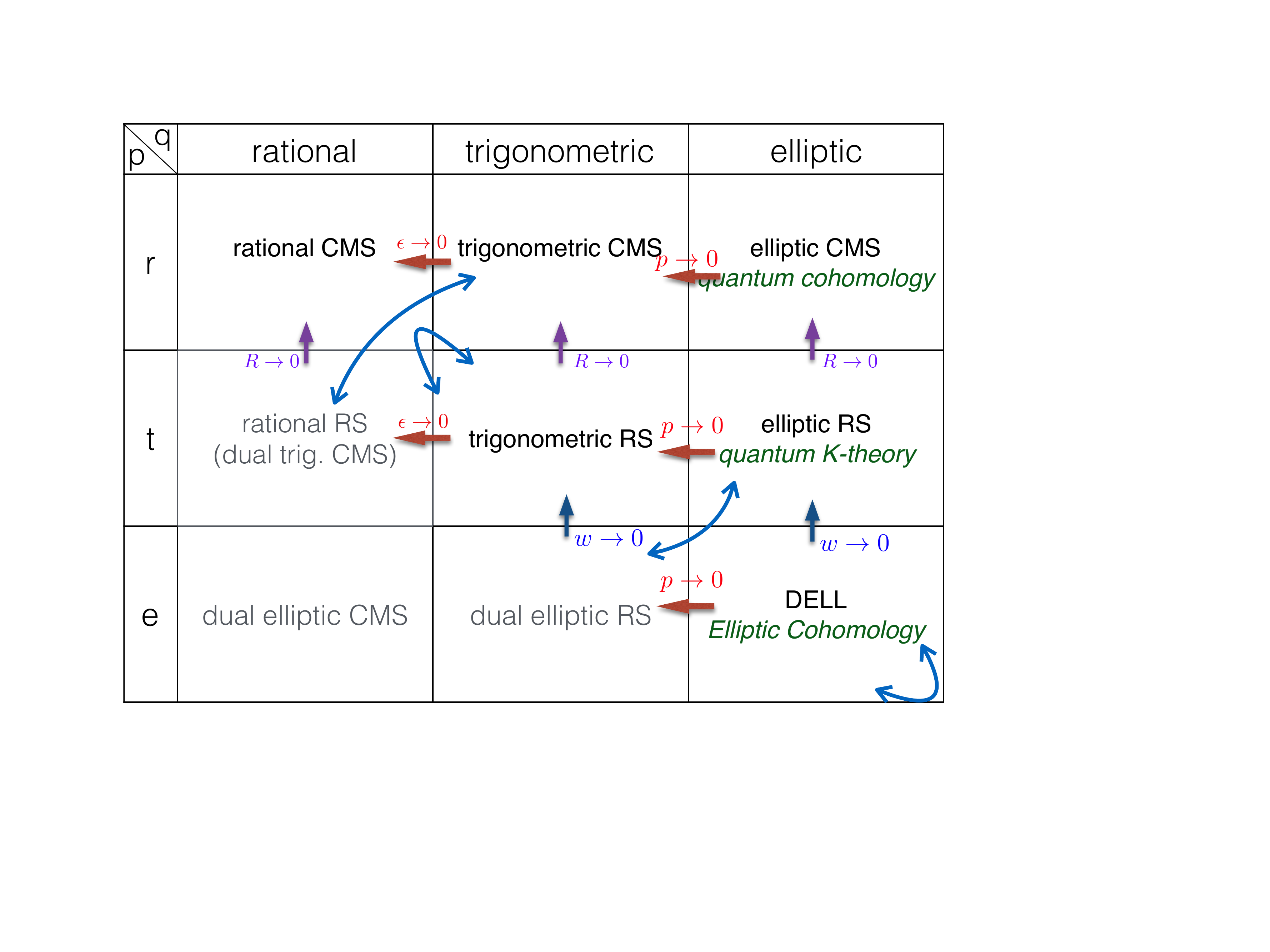}
\caption{The ITEP table \cite{Mironov:2000if} of integrable many body systems according to their periodicity properties in coordinates $q$ (columns) and momenta $p$ (rows) together with their geometric interpretations.}
\label{fig:theiteptable}
\end{figure}

The Calogero and Ruijsenaars families (first and second rows of the table respectively) have a well-established geometric interpretation.

\subsection{Spectrum of the Elliptic RS Model}
In this subsection we would like to prove one of the conjectures of \cite{Koroteev:2019tz}.

\begin{theorem}\label{Conj1}
Let $\textbf{x}=(x_1,\dots x_N)$ be the position vector of the eRS model and $\mathcal{Z}^{{\text{RS}}}(\textbf{a},\textbf{x})=\lim\limits_{w\to 0}\mathcal{Z}^{6d/4d}_{\text{inst}}(w,p,\textbf{x})$ is its wavefunction. Then the following equality holds
\begin{equation}\label{eq:eRSEigenproblemTh}
\mathcal{H}_k \mathcal{Z}^{{\text{RS}}}(\textbf{a},\textbf{x}) = \lambda_k (\textbf{a})\mathcal{Z}^{{\text{RS}}}(\textbf{a},\textbf{x})\,,\qquad k=1,\dots,N-1\,.
\end{equation}
where the eigenvalues read
\begin{equation}
\lambda_k(\textbf{a})=\prod\limits_{n = 0}^{k-1} \frac{\theta(t^{N-n})}{\theta(t^{n+1})} \cdot\frac{\mathcal{Z}^{{\text{RS}}}(\textbf{a},t^{\vec{\rho}}q^{\vec{\omega_k}})}{\mathcal{Z}^{{\text{RS}}}(\textbf{a},t^{\vec{\rho}})}\,,\qquad k=1,\dots, N-1
\end{equation}
where $\vec{\omega_k}$ is the $k$-th fundamental weight of representation of $SU(N)$ and $\vec{\rho}=\big( (N-1)/2, (N-3)/2, \dots, (3-N)/2,(1-N)/2\big)$ is the $SU(N)$ Weyl vector.
\end{theorem}

\begin{proof}
In the above system of difference equations, the Hamiltonians have the form
\begin{equation}
\mathcal{H}_k = \sum\limits_{\substack{\mathcal{I} \subset \{1 \ldots N\} \\ |\mathcal{I}|=k}} \ \prod\limits_{\substack{j \notin \mathcal{I} \\ i \in \mathcal{I}}} \dfrac{\theta( t x_i / x_j)}{\theta(x_i/x_j)} \ \prod\limits_{i \in I} p  _i
\end{equation}
where $p _i$ are the shift operators as defined earlier. While for generic $x_i$ this is a non-trivial system relating values of the eigenfunctions at multiple different points, there exists a specific value of $x_i$ at which all but one terms in the left hand side of each equation in the above system vanish. This is the point $x_i = t^{\rho_i} = t^{(N+1)/2-i}$. Indeed it is easy to see that, in the $k$-th equation of \eqref{eq:eRSEigenproblemTh}, the product of theta functions at $x_i = t^{\rho_i}$ necessarily contains a factor of $\theta(1)$ because of the linear dependence of $i$ in $\rho_i$, unless $I = \{1,2,\ldots,k\}$. Recall from \eqref{eq:ThetaOdd} that $\theta(1) = 0$, therefore the $k$-th equation of the system at $x_i = t^{\rho_i}$ specializes to
\begin{equation}
\prod\limits_{i = 1}^{k} \prod\limits_{j = k+1}^{N} \dfrac{\theta(t^{j-i+1})}{\theta(t^{j-i})} \ \mathcal{Z}^{{\text{RS}}}(\textbf{a},t^{\vec{\rho}}q^{\vec{\omega_k}}) = \lambda_k (\textbf{a})\mathcal{Z}^{{\text{RS}}}(\textbf{a},t^{\vec{\rho}})\,,\qquad k=1,\dots,N-1\,.
\end{equation}
where $\vec{\omega_k} = \big(\underbrace{1, \ldots, 1}_{k}, \underbrace{0, \ldots, 0}_{N-k}\big)$. Simplifying the product in the l.h.s. and dividing by the value of the eigenfunction in the r.h.s., we obtain the desired expression for the eigenvalues.

\end{proof}

\section{Double Inozemtsev Limits of DELL}
The usual Inozemtsev limit \cite{inozemtsev1989} describes a transition from elliptic integrable systems -- Calogero or Ruijsenaars (respectively the first and the second rows of the ITEP table) -- into Toda and q-Toda models correspondingly. More precisely, trigonometric models (second column of the table) are sent to open (q)Toda chains, while elliptic models (third column) become affine (q)Toda chains after applying Inozemtsev  (see \cite{Gorsky:2019yqp} for review). 

Upon the Inozemtsev limit on first scales the coordinates of Calogero or Ruijsenaars particles by a factor which is proportional to the order number of each particle and proportional to the coupling constant. The the coupling constant of the integrable system in question is sent to infinity while the elliptic modulus is sent to zero provided that a certain combination thereof remains finite (see below).

Since we are dealing with the double elliptic model which two elliptic parameters $p$ and $w$ one may try to build a more complex limit which involves scaling of the form $t\to\infty,\, p,w\to 0$ such that certain combination of the above three parameters $F(t,p,w)$ remains constant. We refer to such limits as \textit{double Inozemtsev limits}. 

\vskip.1in
More presicely the double Inozemtsev limit consists of the following rescaling
\begin{equation}\label{eq:scalingparam}
x_a\mapsto t^{-a} x_a,\qquad p_a\mapsto t^{-a-1/2} p_a,\qquad p=t^\alpha \Lambda, \qquad w = t^\beta M\,,
\end{equation}
and taking the limit $t\to 0$. Here $\Lambda$ and $M$ are nonzero constants. We study this limit for different values of $\alpha$ and $\beta$.

\subsection{{\bf $\alpha=N$, $\beta>1$}}
In this case, due to the quadratic dependence of the Hamiltonians on $w$, higher elliptic corrections vanish and the model becomes equivalent to the elliptic Ruijsenaars-Schneider model. Thus its Inozemtsev limit is the quantum affine q-Toda system \cite{Gorsky:2019yqp}. Its first Hamiltonian reads
\begin{equation}\label{eq:ElToda}
H^{\text{aff q-Toda}}_1=\sum_{i=2}^N \left(1-\frac{x_{i-1}}{x_i}\right)p_i +  \left(1-\Lambda\frac{x_N}{x_1}\right)p_1\,,
\end{equation}
Upon this limit $\mathcal{O}_0=1$.

\subsection{{\bf $\alpha=N$, $\beta=1$}}
This is a new limit which at first glance yields a new elliptic model. The normal ordered DELL operators $\mathcal{O}_k$ in this limit read as follows
\begin{equation}\label{eq:Oklim}
\widetilde{\mathcal{O}}_k = c_N : \sum_{i_1<\dots<i_k}\prod\limits_{a < b} \theta\left( M^{b-a+m_b - m_a} \frac{x_a}{x_b}\frac{p_b}{p_a}\Big\vert \widetilde p \right) :\ p_{i_1}\cdots p_{i_k}\,,
\end{equation}
where $\widetilde  p = \Lambda M^N$ and $m_a = \delta_{a \in I}$, and $c_N$ is a certain constant. However, the above Hamiltonians can undergo the following chain of transformations. First, let $x_a\mapsto M^{-a}x_a$, then we perform the $SL(2;\mathbb{Z})$ $T$-transformation which in the multiplicative form reads
\begin{equation}
x_a\mapsto x_a p_a, \,\qquad \qquad p_a\mapsto p_a\,.
\end{equation}
This transformation is none other than conjugation of the operators by the Gaussian term $\exp( \partial^2 /2 )$ (recall that the shift operator is $p = q^\partial$). After that \eqref{eq:Oklim} becomes
\begin{equation}\label{eq:Oklim1}
\widetilde{\mathcal{O}}_k = c_N \sum_{i_1<\dots<i_k}\prod\limits_{a < b} \theta\left( M^{m_b - m_a} \frac{x_a}{x_b}\Big\vert\widetilde  p \right) p_{i_1}\cdots p_{i_k}\,.
\end{equation}
According to \cite{Koroteev:2019tz} the above Hamiltonians lead to the elliptic Ruijsenaars-Schneider model $\mathcal{H}^{\text{eRS}}_k=\widetilde{\mathcal{O}}_0^{-1}\widetilde{\mathcal{O}}_k$ with coupling constant $M$ and elliptic parameter $\widetilde p$.

\vskip.1in

The above limit can be illustrated using the following brane description (see \figref{fig:quiver6d1}). Here vertically we have NS5 branes and horizontally we find D5 branes. The starting point is the little string theory with periodic NS5 direction which can be then scaled to the field theory limit by sending the vacuum expectation of the dilaton to infinity.
\begin{figure}[h]\label{fig:quiver6d1}
\includegraphics[scale=0.5]{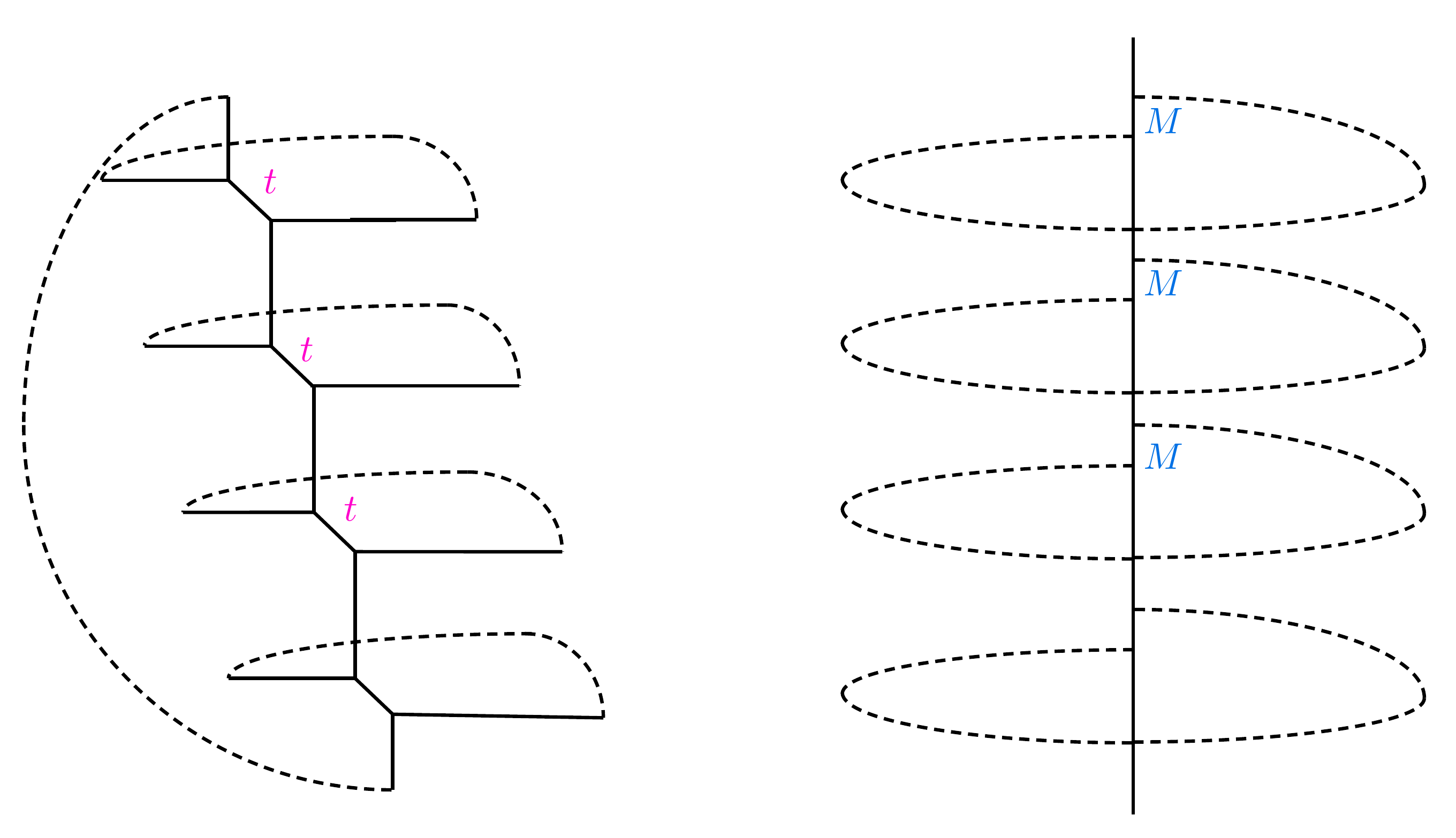}
\caption{Brane description of the 6d theory whose infrared regime is described by the DELL model with coupling $t$ (left) and the 5d theory whose Seiberg-Witten solution gives rise to the eRS model with coupling $M$ (right).}
\end{figure}
The left figure describes 6d $\mathcal{N}=(1,0)^*$ theory on $\mathbb{R}^4\times S^1_R\times S^1_{R'}$ which can be viewed as a 5d $\mathcal{N}=1^*$ theory on $\mathbb{R}^4\times S^1_{R'}$ whose Yang-Mills coupling depends on the compactification radius $g_{\text{YM}}^{-2}=1/R$ along the sixth direction. The complexified 5d gauge coupling is equal to $p\sim e^{-4\pi R'/R}$ while the elliptic modulus of the compactification torus is $w$. The $\mathcal{N}=2^*$ supersymmetry is softly broken to $\mathcal{N}=1^*$ by introducing the mass parameter to the adjoint hypermultiplet $t=e^{R' m}$. 

In the Inozemtsev limit $m\to\infty$ and $R\to 0$ such that $p t^N=\Lambda$ or, equivalently, $N m \sim 1/R$. This results in collapsing of one of the periodic directions in the 6d description. Then we apply a T-duality transformation along the compact NS5 brane direction which results in the NS5 brane with infinite radius as is shown on the right side of \figref{fig:quiver6d1}. Interestingly, during these chain of limits and transformations the former elliptic modulus parameter $w$ scalers to the new mass parameter $M$ while the 5d Yang-Mills gauge coupling of the limiting theory becomes $\widetilde  p = \Lambda M^N$.

\subsection{Fractional Limits}
Let us consider a different scaling
\begin{equation}\label{eq:scalingparam}
x_a\mapsto t^{-\frac{a}{\gamma}} x_a,\qquad p_a\mapsto t^{-\frac{a}{\gamma}-\frac{1}{2\gamma}} p_a,\qquad p=t^\alpha \Lambda, \qquad w = t^\beta M\,.
\end{equation}
Note that this scaling is not equivalent to the previous family via redefinition $u=t^{1/2}$ since the dependence on $w$ of DELL Hamiltonians is nonlinear. This scaling, as we shall see below, leads to q-Toda-like systems, albeit with certain restrictions.

\vskip.1in
Consider an element in the product inside theta-function
\begin{equation}
\theta(z|p)\supset \prod_{k=0}^\infty(1-p^k z)
\end{equation}
where in our case $z=t^{n_i-n_j}\frac{x_i}{x_j}$. After the rescaling the power of $t$ in the $k$-th term in the product  will become
\begin{equation}
P_{ij}=\alpha k + n_i - n_j + \frac{j-i}{\gamma}
\end{equation}
Our goal is to find the most singular terms in each theta function ratio thus we need to make sure that the above exponent is negative (or zero). This singular contributions will be compensated by $w$ in the expression for the DELL Hamiltonian. Consider the $w$-terms in more detail.
\begin{itemize}
\item $w^0$: This corresponds to vector $n_i=0$ or $n_i=1$ since the function $n(n-1)/2$ vanishes at $0$ and $1$ and otherwise is positively defined at integer $n$.

\item $w^1$: Notice that $n(n-1)/2=1$ for $n=-1$ and $n=2$. However, we need to satisfy the constraint $\sum n_i = k$ for $\mathcal{O}_k$. Consider $k=1$ for now. Thus we need $\sum_i n_i(n_i-1)/2=1$ as well. One can see that this can be achieved by having vector $\textbf{n}=(1,-1,1,0,\dots,0)$ and its permutations.
Therefore the values of $n_i-n_j$ in $P_{ij}$ range over $-2,-1,0,1,2$. In order to obtain singular terms we need the whole expression to be negative so $\gamma$ can be chosen in such a way that
$n_i - n_j + \frac{j-i}{\gamma}\leq 0$ or
\begin{equation}
\gamma \geq \frac{j-i}{2}
\end{equation}
In other words $\gamma = \frac{j-i}{2}$ defines the range of the interaction between the particles (the interaction becomes infinitely weak beyond this range). Once $\gamma$ is set we can adjust $\alpha$. For instance, $\gamma = 2, \alpha = 3/2, \beta=3/2$ will provide one term proportional to $M$ in the expression for $\mathcal{O}_1$.

\end{itemize}

\vskip.2in

The above analysis leads to the following operators for two and three particles.

\paragraph{{\bf Two Particles}}
Consider $\alpha=\beta=\frac{3}{2}$ in \eqref{eq:scalingparam}, then
\begin{align}
\mathcal{O}_0&=1-\Lambda  M\frac{ x_2}{x_1}\frac{p_1}{p_2} \,,\cr
\mathcal{O}_1&=  \left(1-\Lambda\frac{x_2}{x_1}+\Lambda ^3 M^2\frac{x_2^2}{x_1^2}\frac{ p_1 }{p_2}\right)p_1-\frac{x_1}{x_2}p_2 \,,
\end{align}

\paragraph{{\bf Three Particles}}
Consider $\alpha=\beta=2$, then
\begin{align}
\mathcal{O}_0&=1-\Lambda  M\frac{ x_3}{ x_1}\frac{ p_1}{p_3}\,,\cr
\mathcal{O}_1&= \left(1-\Lambda \frac{x_3}{x_1}\right)p_1-\frac{ x_1}{x_2}p_2+\Lambda ^2 M \frac{x_3^2}{x_1 x_2}\frac{p_1 p_2 }{p_3 }\,,\cr
\mathcal{O}_2&= \left(p_2-\Lambda\frac{
   x_3}{x_1}p_2-\frac{x_2}{x_3}p_3 \right)p_1+\Lambda ^2 M\frac{x_2 x_3}{x_1^2} p_1^2
\end{align}

However, for higher number of particles this limit produces Hamiltonians which only involve interaction with a fewer number of particles and completely ignore interactions with the others. Notice also that the above Hamiltonians are deformation of finite q-Toda chains (not affine).

\vskip.1in
Therefore at this moment we cannot find a novel limit of DELL with $N>3$ particles which would present a novel integrable system which lies `in between' eRS and DELL in the hierarchy. Most likely such model does not exist.

One may still be able to find a different scaling other than \eqref{eq:scalingparam} of the form
\begin{equation}
x_a\mapsto t^{f(a)} x_a,\qquad p_a\mapsto t^{g(a)} p_a,\qquad p=t^\alpha \Lambda, \qquad w = t^\beta M\,,
\end{equation}
where $f(a)$ and $g(a)$ are certain functions, in order to capture more terms from the DELL Hamiltonians in the limit. We encourage the reader to try.

\section{DELL System in the Dual Frame}
In this section we would like to comment on the recent results of \cite{Mironov:2021tx}. In particular, the authors conjecture the existence of a pq-self-dual DELL system, which, according to the reasons listed below, does not exist without invoking spin degrees of freedom (see
Fig \ref{fig:spindell} for the toric diagram).
\begin{figure}[h]\label{fig:spindell}
\includegraphics[scale=0.8]{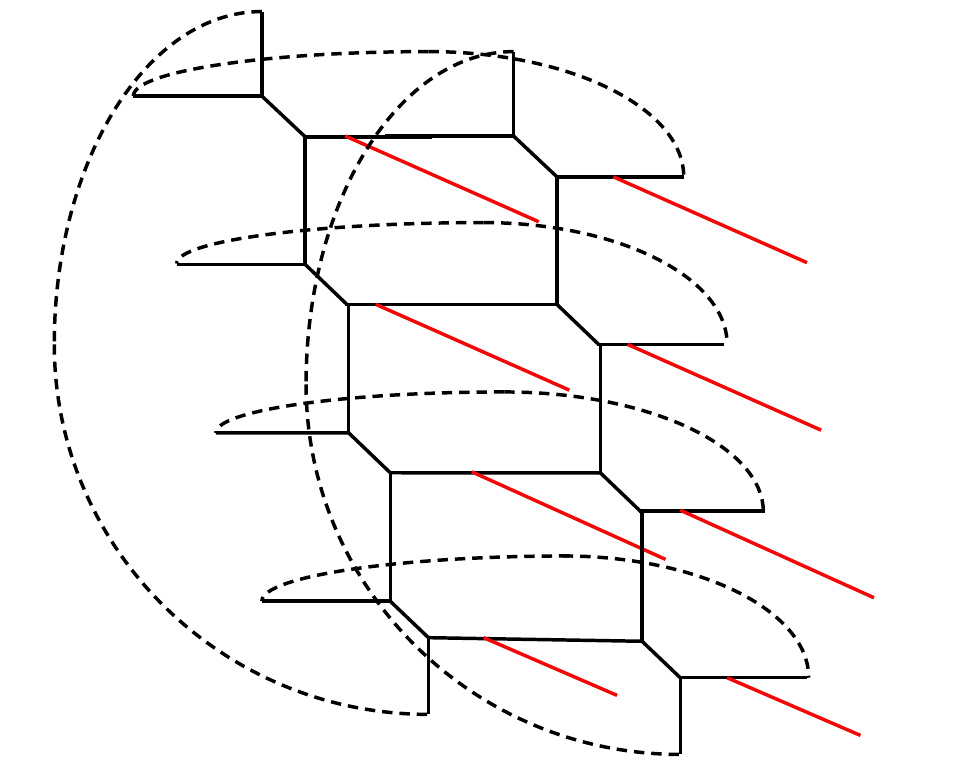} \quad  \includegraphics[scale=0.9]{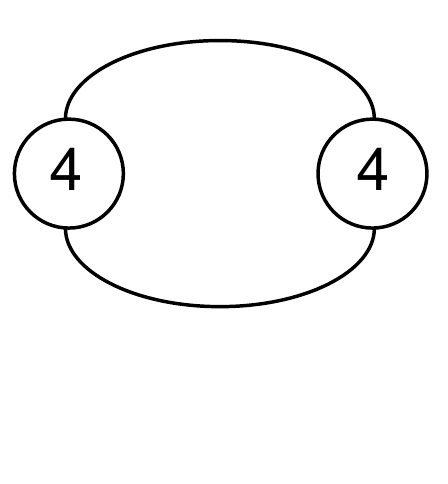}
\caption{Toric diagram for affine $\widehat{A}_1$ 6d quiver with two $U(3)$ nodes whose Seiberg-Witten geometry is described by spin-$\frac{1}{2}$ DELL system with four particles.}
\end{figure}

The pq-duality is manifested differently in physical setting in various dimensions. In 5d it is realized as the S-duality of the pq-brane web, while in 3d it is manifested by the 3d mirror symmetry. Since the description of the DELL model uses periodic pq-webs the pq-duality should also be manifested via the S-duality (90-degree rotation of the brane web). For instance, for the theory in \figref{fig:spindell} the pq-dual quiver is the following
\begin{figure}[h]
\includegraphics[scale=0.7]{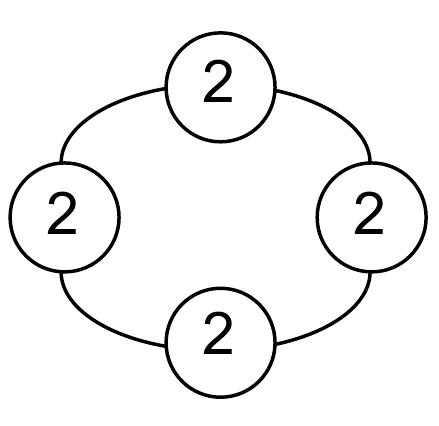}
\end{figure}
which yields spin-$\frac{3}{2}$ DELL model with two particles.

\vskip.1in

Every Hamiltonian of the Calogero-Ruijsennars-DELL family has two geometric/physics interpretations -- the so-called \textit{`electric frame'} (or $a$-frame) and the \textit{`magnetic frame'} (or $x$-frame). The two frames are related via the pq-duality which interchanges $a$ and $x$ variables together with the involution on the coupling constant parameter $\hbar\mapsto\hbar^{-1}$.

In addition to the pq-duality described above, in both electric and magnetic frames (see Fig. \ref{fig:dualityweb}) there are Fourier transforms which interchange momenta and coordinates in the Hamiltonians.

\begin{figure}[h]
\includegraphics[scale=0.6]{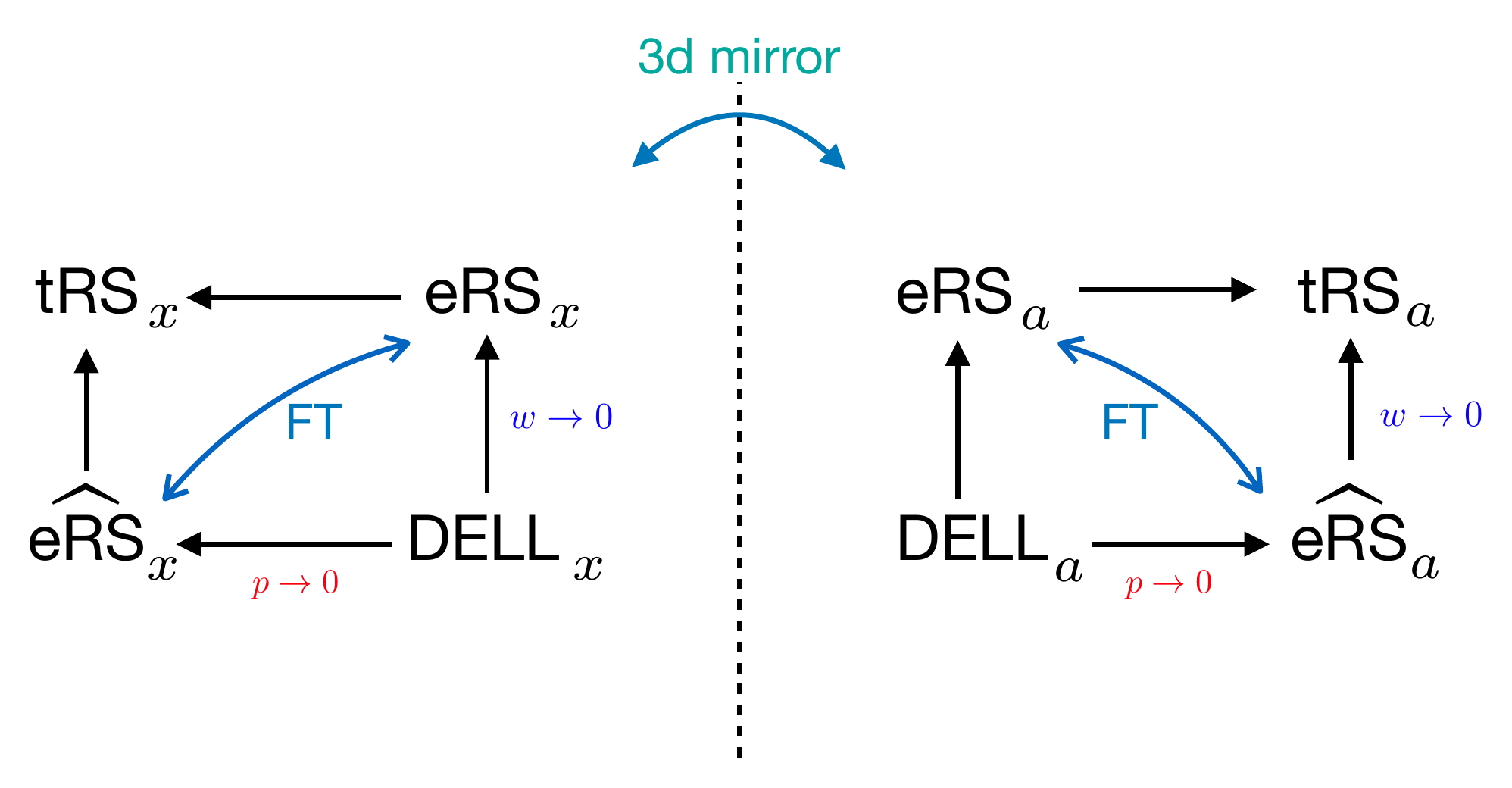}
\caption{Network of degeneration limits from the DELL models (up to the tRS models) whose Hamiltonians act on $x$-variables (DELL$_x$) and on $a$-variables (DELL$_a$). The two families (two dual ITEP tables, see Fig. \ref{fig:theiteptable}) are related via 3d mirror symmetry which interchanges parameters $x$ and $a$.}
\label{fig:dualityweb}
\end{figure}

Therefore the only self-dual family of double elliptic integrable models are spin-DELL systems whit spin-$\frac{N-1}{2}$ and $N$ particles, since their toric diagrams are pq-invariant. These models describe Seiberg-Witten geometry of affine $\widehat{A}_{N}$ quiver gauge theory with $U(N)$ gauge groups by means of the following generating function
\begin{align}
\widehat{ \mathcal{O^S}}(z) \ &= \ \sum\limits_{n \in {\mathbb Z}} \ \widehat{ \mathcal{O^S}}_n \ z^n  \notag\\
&=\sum\limits_{a=1}^N\sum\limits_{\textbf{n}_{a} \in \mathbb{Z}^N} \ (-z)^{\sum n_{a,i}} \ w^{\sum \frac{n_{a,i}(n_{a,i} - 1)}{2}} \ \prod\limits_{a=1}^N\prod\limits_{i < j} \theta\left( t^{n_{a,i} - n_{a,j}} \frac{x_{a,i}}{x_{a,j}} \Big\vert p\right) \ p_{a,1}^{n_{a,1}} \ldots p_{a,N}^{n_{a,N}}\,.
\label{eq:CurrentOzSpin}
\end{align}
We expect to study the properties of spin-DELL system in the near future.

Below we shall describe electric and magnetic frames for Ruijsenaars models in more details.

\subsection{3d Mirror Symmetry and pq-Duality in tRS models}
First, we remind the reader about the solutions of quantum trigonometric Ruijsenaars-Schneider model in electric and magnetic frames using quantum geometry of quiver varieties. At generic values of the equivariant parameters ${\bf a}$ the eigenfunctions are series. As we shall review below at certain resonance values \eqref{eq:truncationa} these series truncate to Macdonald polynomials which are most commonly discussed in the representation theory literature. The Pieri rule and its elliptic generalization, which were discussed in \cite{Mironov:2021tx}, will therefore follow from the tRS eigenvalue problem in the electric frame \eqref{eq:tRSEigen}, after the resonance conditions have been imposed.

\subsubsection{Magnetic Frame}
In the magnetic frame the tRS operators act on the K\"ahler parameters $\zeta_i$ of the quiver variety (in physics language on the Fayet-Iliopoulos couplings, see \cite{Bullimore:2015fr} for gauge theory interpretation of electric and magnetic frames).
\begin{definition}
The difference operators of trigonometric Ruijsenaars-Schneider model are given by
\begin{equation}
T_r(\boldsymbol{\zeta})=\sum_{\substack{\mathcal{I}\subset\{1,\dots,n\} \\ |\mathcal{I}|=r}}\prod_{\substack{i\in\mathcal{I} \\ j\notin\mathcal{I}}}\frac{\hbar\,\zeta_i-\zeta_j}{\zeta_i-\zeta_j}\prod\limits_{i\in\mathcal{I}}p_k \,,
\label{eq:tRSRelationsEl1}
\end{equation}
where $\boldsymbol{\zeta}=\{\zeta_1,\dots, \zeta_n\}$, the shift operator $p_k f(\zeta_k)=f(q\zeta_k)$.
\end{definition}

The spectrum of the tRS operators is described geometrically by K-theory vertex function for the cotangent bundle to the complete flag variety $T^*\mathbb{F}l_n$ \cite{Koroteev:2018a, Koroteev:2018cmp} and physically as vortex partition functions of the $T[U(n)]$ gauge theory on $\mathbb{R}^2\times S^1$ \cite{Bullimore:2015fr}. We remind the reader the definitions and theorems.

\begin{theorem}[\cite{Koroteev:2018cmp} Thm. 2.10]\label{Th:tRSeqproof}
Let $V^{(1)}_{\textbf{p}}$ be the coefficient for the vertex function for the cotangent bundle to the complete flag variety $T^*\mathbb{F}l_n$. Define
\begin{equation}
\mathsf{V}^{(1)}_{\textbf{p}}= \prod\limits_{i=1}^n\frac{\theta(\hbar^{i-n}\zeta_i,q)}{\theta(a_i\zeta_i,q)}\cdot V^{(1)}_{\textbf{p}}\,,
\label{eq:DefVvertMac}
\end{equation}
where $\theta(x,q)=(x,q)_\infty (qx^{-1},q)_\infty$ is basic theta-function and
\begin{equation}
V^{(1)}_{\textbf{p}}(z) = \sum\limits_{d_{i,j}\in C} \prod_{i=1}^{n-1} \left(t\frac{\zeta_{i}}{\zeta_{i+1}}\right)^{d_i} \prod\limits_{j,k=1}^{i}\frac{\left(q\frac{x_{i,j}}{x_{i,k}},q\right)_{d_{i,j}-d_{i,k}}}{\left(\hbar\frac{x_{i,j}}{x_{i,k}},q\right)_{d_{i,j}-d_{i,k}}}\cdot\prod_{j=1}^{i}\prod_{k=1}^{i+1}\frac{\left(\hbar\frac{x_{i+1,k}}{x_{i,j}},q\right)_{d_{i,j}-d_{i+1,k}}}{\left(q\frac{x_{i+1,k}}{x_{i,j}},q\right)_{d_{i,j}-d_{i+1,k}}}\,,
\label{eq:V1pdef}
\end{equation}
are the coefficient functions for the vertex function for fixed point ${\bf p}$.
Then $\mathsf{V}_{\textbf{p}}$ are eigenfunctions for tRS difference operators \eqref{eq:tRSRelationsEl1} for all fixed points $\textbf{p}$
\begin{equation}
T_r(\boldsymbol{\zeta}) \mathsf{V}^{(1)}_{\textbf{p}} = e_r (\mathbf{a}) \mathsf{V}^{(1)}_{\textbf{p}}\,, \qquad r=1,\dots, n\,,
\label{eq:tRSEigenz}
\end{equation}
where $e_r$ is elementary symmetric polynomial of degree $r$ of $a_1,\dots, a_n$\,.
\end{theorem}

The following statement describes truncation of the eigenfunctions $\mathsf{V}^{(1)}_{\textbf{p}}$ from the above theorem. In the presence of certain resonance conditions these series become Macdonald polynomials.

\begin{proposition}[\cite{Koroteev:2018cmp} Prop. 2.11]
\label{Prop:Truncation}
Consider coefficient functions for K-theory of $\textsf{QM}$ to $X_n$  \eqref{eq:DefVvertMac} for all fixed points of the maximal torus. Let $\lambda$ be a partition of $k$ elements of length $n$ and $\lambda_1\geq\dots\geq\lambda_n$. Let
\begin{equation}
\frac{a_{i+1}}{a_{i}}=q^{\ell_{i}}\hbar\,,\quad \ell_i = \lambda_{i+1}-\lambda_{i}\,, \quad i=1,\dots,n-1\,.
\label{eq:truncationa}
\end{equation}
Then there exists a fixed point $\textbf{q}$ for which
\begin{equation}
\mathsf{V}_{\textbf{q}} = P_\lambda(\boldsymbol{\zeta};q,\hbar)\,.
\end{equation}
\end{proposition}

For instance, for $T^*\mathbb{P}^1$ the condition \eqref{eq:truncationa} reads $a_1=a q^{\lambda_1}\hbar^{-1}\,,a_2= a q^{\lambda_2}$. Then we have the following Macdonald polynomials for some simple tableaux.
\begin{align}
\mathrm{P}_{{\tiny\yng(1)}}&=\zeta _1+\zeta _2,\cr
\mathrm{P}_{{\tiny\yng(1,1)}}&=\zeta _1^2+\zeta _2^2+\frac{ (q+1) (\hbar -1)}{q \hbar -1}\zeta _1 \zeta _2,\cr
\mathrm{P}_{{\tiny\yng(1,1,1)}}&=\zeta _1^3+\zeta
   _2^3+\frac{ \left(q^2+q+1\right) (\hbar -1)}{q^2 \hbar -1}\zeta _2 \zeta _1^2+\frac{
   \left(q^2+q+1\right) (\hbar -1)}{q^2 \hbar -1}\zeta _2^2 \zeta _1\,.
\end{align}

\subsubsection{Electric Frame}
In the $a$-frame the tRS Hamiltonians read
\begin{equation}
T_r(\textbf{a})=\sum_{\substack{\mathcal{I}\subset\{1,\dots,n\} \\ |\mathcal{I}|=r}}\prod_{\substack{i\in\mathcal{I} \\ j\notin\mathcal{I}}}\frac{t\,a_i-a_j}{a_i-a_j}\prod\limits_{i\in\mathcal{I}}p_i \,,
\label{eq:tRSRelationsEl}
\end{equation}
where $\textbf{a}=\{a_1,\dots, a_{n}\}$, the shift operator $p_i f(a_i)=f(qa_i)$ and we denoted $t=\frac{q}{\hbar}$.

\begin{theorem}\label{Th:tRSEigen}
The following function constructed for the cotangent bundle to the complete flag variety
\begin{equation}
\mathrm{V}(\textbf{a},{\boldsymbol \zeta}) = \frac{e^{\frac{\log \zeta_n \sum_{i=1}^{n-1}\log a_i}{\log q}}}{2\pi i} \int\limits_{C} \prod_{m=1}^{n-1}\prod_{i=1}^{m} \frac{ds_{m,i}}{s_{m,i}} E(s_{m,i})\,\, e^{-\frac{\log \zeta_{m}/\zeta_{m+1} \cdot\log s_{m,i}}{\log q}} \cdot \prod_{j=1}^{\textbf{v}_{m+1}}H_{m,m+1}\left(s_{m,i},s_{m+1,j}\right)\,,
\label{eq:GenerictRSSolution}
 \end{equation}
where $E(s_{m,i})$ and $H_{m,m+1}$ are certain rational functions and contour $C$ is chosen in such a way that shifts of the contour $\textbf{s}\to q^{\pm 1} \textbf{s}$ do not encounter any poles,
satisfies tRS difference relations
\begin{equation}
T_r(\textbf{a}) \mathrm{V}(\textbf{a},{\boldsymbol \zeta}) = e_r ({\boldsymbol \zeta}) \mathrm{V}(\textbf{a},{\boldsymbol \zeta})\,, \qquad r=1,\dots, n
\label{eq:tRSEigen}
\end{equation}
where function $e_r $ is $r$th elementary symmetric polynomial $\zeta_1,\dots,\zeta_n$.
\end{theorem}

\vskip.1in

The 3d mirror symmetry/pq-duality states that for the cotangent bundle to the complete flag variety the eigenfunctions and eigenvalues are related to each other via the interchange of the K\"ahler and equivariant parameters. Namely, for a certain fixed point of the maximal torus $\textbf{p}$ we have
\begin{equation}
V_{\textbf{p}}({\bf a},{\boldsymbol \zeta},\hbar) = V_{\textbf{p}}({\boldsymbol \zeta},{\bf a},q/\hbar)\,,
\end{equation}
and the tRS operators \eqref{eq:tRSRelationsEl1} and \eqref{eq:tRSRelationsEl} are related in the same way.

After imposing the resonance conditions \eqref{eq:truncationa} the tRS eigenvalue relations in the electric frame \eqref{eq:tRSEigen} turns into the Pieri rules (upon applying a conjugation to the tRS Hamiltonians) if we look at it from right to left. Consider the relation for $r=1$
\begin{equation}
\mathrm{P}_{{\tiny \yng(1)}} (q^{\lambda_{i}}\hbar^{n-i})\cdot \mathrm{P}_\lambda({\boldsymbol \zeta})= \sum_{i=1}^n\prod_{j\neq i}\frac{q^{\lambda_i-\lambda_j+1}\hbar^{j-i-1}-1}{q^{\lambda_i-\lambda_j}\hbar^{j-i}-1} \mathrm{P}_{\lambda+{\tiny \yng(1)\, }_i}({\boldsymbol \zeta})  \,,
\label{eq:tRSEigenPieri}
\end{equation}
where $\lambda+{\tiny \yng(1)\, }_i$ denotes a diagram obtained from the diagram $\lambda$ by adding one box in column $i$.

\subsection{PQ-duality in eRS models}
Analogously one can explain the pq-duality between RS$_a$  and $\widehat{\text{eRS}}_x$ from \figref{fig:dualityweb}. In \cite{Mironov:2021tx} the former eigenvalue problem was presented as an elliptic Pieri rule which can be understood as truncation of the eigenvector of $H^a_{eRS}$ which geometrically is the equivariant Euler characteristic of the affine Laumon space. 
The elliptic Pieri rule reads
\begin{equation}
\mathrm{E}_{{\tiny \yng(1)}} (p_k)\cdot \mathrm{E}_\lambda(p_k)= \sum_{i=1}^{l(\lambda)+1}\prod_{j\neq i}\psi\left(\frac{q^{\lambda_i}\hbar^{i-1}}{q^{\lambda_i}\hbar^{j-1}}\right) \mathrm{E}_{\lambda+{\tiny \yng(1)\, }_i}(p_k)  \,,
\label{eq:EllipticPieri}
\end{equation}
where $p_k=\sum_i z_i^k$ is the $k$-th power symmetric polynomial and function
$$
\psi(x)=\frac{\theta(q\hbar x|w)\theta(\hbar^{-1}x|w)}{\theta(x|w)\theta(q x|w)}
$$
in the limit $w\to 0$ reduces to the fraction in the right hand side of \eqref{eq:tRSEigenPieri} with an extra factor which is and artifact of different normalization of Ruijsenaars Hamiltonians in \cite{Mironov:2021tx}. 

Meanwhile the $\widehat{\text{eRS}}_x$ in \cite{Mironov:2021tx} was referred to as the eigenfunction of the DELL Hamiltonians in the elliptic-trigonometric limit ($w\to 0$).

\section{DELL and Separated Variables}
 One more approach to the Inozemtsev limit for DELL is based on the separation of variables -- the procedure 
effective for the interacting many-body systems. The system of $N$ interacting degrees of freedom
with a  simple phase-space manifold gets mapped via canonical transformation into the system 
of $N$ non-interacting degrees of freedom with complicated phase-space manifold. The 
derivation of the separated variables and the corresponding canonical transformation is
not a simple issue and it was identified in the brane terms \cite{separation} as the representation of the 
integrable model associated with the geometry  involving  several types of branes 
into a system of D0 branes on some
manifold via Fourier-Mukai transforms. More formally we map the initial phase space of the $N$-body holomorphic system $M$ into 
the phase space identified with Hilbert scheme of $N$ points on some two-dimensional complex manifold $X$: Hilb$^N(X)$.
This manifold can be considered as the mirror to some two dimensional Calabi-Yau manifold and the prepotential
attributed to the integrable system is identified with the partition function of the topological string on $X$.

Such approach has been applied for the two-body DELL system in \cite{Braden:2001yc}. It was argued in \cite{Braden:2001yc} that the two-body case the manifold $X$ is the two-torus bundled over another two-torus, where one torus modulus corresponds to the  
gauge coupling while the second corresponds to the elliptic modulus of the compactification torus. The Inozemtsev limit
in terms of the separation of variables corresponds to the peculiar degeneration of the manifold $X$.
The clear-cut example of resulting geometry for two-body holomorphic periodic q-Toda model can be read off from 
\cite{Hatsuda:2016mdw}. It was shown there that the energy level for the two-body system  
coincides with the mirror curve for the $\mathbb{P}_1\times \mathbb{P}_1$ manifold which in this case plays
the role of manifold $X$. In our paper we have discussed at the algebraic level the
possible intermediate degenerations of  two-torus bundled over two-torus manifold into $\mathbb{P}_1\times \mathbb{P}_1$.

Let us remark that it was argued in \cite{Hatsuda:2016mdw} that the periodic q-Toda Hamiltonian 
describes the particle on the lattice in the magnetic field where the lattice is modeled
by the cosine potential. The spectrum of this model enjoys the Hofstadter butterfly pattern.
However the generalization of such magnetic interpretation for interacting N-body system
in magnetic field is more complicated. For instance the Lauphlin wave function in the
continuum case is related with the wave function of  rational Calogero model in the external oscillatory potential 
which can be mapped into the trigonometric Calogero model \cite{Nekrasov:1997jf} describing the particles 
on the cylinder when the magnetic field gets mapped into the radius of the cylinder. It is this interacting system to which 
the Inozemtsev limit can be  applied and upon
the proper rescaling 
the trigonometric Calogero model reduces  into the open Toda chain. However it is not clear how this simple 
model of interacting particles in magnetic field
can be extended to the  elliptic case.

\section*{Acknowledgments}
The work of AG was supported by Basis Foundation grant 20-1-1-23-1 and by grant
RFBR-19-02-00214. PK is partially supported by AMS Simons Travel Grant.

\bibliography{cpn1}
\end{document}